\documentclass{ifacconf}

\usepackage{graphicx}      
\usepackage{natbib}        
\usepackage{graphicx}
\usepackage{epstopdf}
\usepackage{subfig}
\usepackage{epsfig}
\usepackage{amsfonts}
\usepackage{graphicx}
\usepackage{epstopdf}
\usepackage{epsfig}
\usepackage{amsmath}
\usepackage{amsfonts}
\newtheorem{proof}{Proof}

\begin{document}
\begin{frontmatter}

\title{Local module identification in dynamic networks: do more inputs guarantee smaller variance? \thanksref{footnoteinfo}
  }

\thanks[footnoteinfo]{First version completed on 21 December 2017. Revised 27 April 2018. This work has received funding from the European Research Council (ERC), Advanced Research Grant SYSDYNET, under the European Union’s Horizon 2020 research and innovation programme (grant agreement No 694504).}

\author[First]{M. Mohsin Siraj},
\author[Second]{M.G. Potters} and
\author[First]{Paul M.J. Van den Hof}

\address[First]{Department of Electrical Engineering, 
Eindhoven University of Technology (e-mails: m.m.siraj, p.m.j.vandenhof@tue.nl).}
\address[Second]{Afdeling VolkerData, VolkerRail, Lange Dreef 7, 3131 NJ Vianen (e-mail: max.potters@volkerrail.nl)}

\begin{abstract}                
Recent developments in science and engineering have motivated control systems to be considered as interconnected and networked systems. From a system identification point of view, modelling of a local module in such a structured system is a relevant and interesting problem. This work focuses on the quality, in terms of variance, of an estimate of a local module. We analyse which predictor input signals are relevant and contribute to variance reduction,  while still guaranteeing the consistency of the estimate. For a targeted local module, a comparison of its estimate variance is made between a full-MISO approach and an immersed network setting, where a reduced number of inputs is used, while still guaranteeing consistency. A case study of a four-node network is considered and it is shown that a smaller set of predictor inputs can, under some conditions, result in a smaller variance compared to the full-MISO approach.
\end{abstract}


\end{frontmatter}

\section{Introduction}
Data-driven modelling of structured interconnections of dynamic systems (modules), also known as dynamic networks, is gaining considerable attention, see e.g., \cite{gonccalves2008necessary}, \cite{materassi2010topological}, \cite{chiuso2012bayesian}. These interconnected dynamic networks can be found in many scientific and engineering fields such as power systems, biological systems, etc. One of the challenging problems in modelling of these dynamic networks is the identification of a local module when the structure of the network (topology) is known.

In \cite{van2013identification} the closed-loop prediction error framework, including direct and two-stage identification methods, has been extended to the case of dynamic networks and an analysis in terms of consistency properties has been presented. For the case of a known structure in the network, the situation has been considered that \textit{all} signals that directly map into the output node of the target module are taken into account as predictor inputs (full-MISO approach).
This condition is relaxed in \cite{dankers2016identification}, where an \textit{immersed network} is formed using graph theory tools, resulting in a reduced number of signals to be used as predictor inputs, to achieve consistency of a targeted module estimate.

The selection of inputs not only affects the consistency of the estimate, but also strongly influences its variance. For a large network with many possible sets of predictor inputs, it becomes relevant to analyse which predictor input set results in the smallest variance. In \cite{gevers2006identification}, a variance analyses of a multi-input single-output system with independent inputs (no correlation) is discussed and it has been shown that a change of the experiment by adding an input signal either increases the accuracy of the estimate or at least does not adversely affect it. In \cite{ramazi2014variance}, a variance analysis is provided for spatially correlated inputs in a MISO system. For interconnected systems, a variance analysis in a cascade structure has been studied e.g., in \cite{wahlberg2009variance}, \cite{everitt2017module}. A variance reduction technique in dynamic networks has been presented in \cite{Gunes14}. 

This work focusses on variance analysis in dynamic networks. We consider the identification and dynamic network setup of (\cite{van2013identification}). Specifically, we compare the variance of a targeted module estimate of the full-MISO approach to that of its immersed network counterpart (\cite{dankers2016identification}), which has a reduced number of inputs but increased complexity of some of its modules. Our research contrasts the variance analysis as done e.g., in \cite{gevers2006identification} in the sense that we keep the experiment fixed, but vary over the choice of predictor input signals out of the set of available signals, thereby focussing on the choice of available signals in the  \textit{modelling} process. We consider the research question: what is the benefit of using an additional predictor input signal on the variance of a target module of interest in a dynamic network with known topology?
We also investigate if it is always preferable to include all information (inputs) in the modelling process. To this end, we consider different modelling processes with various (possible) sets of inputs while the experiment setup remains the same.

Our analysis is primarily based on the observations that for the immersed network settings, with a reduced number of inputs, the number of to-be-identified transfer functions is smaller compared to the full-MISO approach. This may result in a smaller variance of a targeted module in the immersed network settings. On the other hand, process noise at an input node, when that input node is removed from the set of predictor inputs, affects the output directly and can deteriorate the signal-to-noise ratio at the output node. In the full-MISO approach the process noise can play the role of an excitation signal for the predictor input node. Therefore, we investigate the tradeoff between the number and complexity of to-be-estimated transfer functions and the  signal-to-noise ratio in the immersed network.

Our approach is to use the parameter covariance (\cite{ljung1998system}) and the asymptotic frequency response variance expressions (\cite{ljung1985asymptotic}) for a targeted module in both the full MISO and the immersed network settings. We derive conditions under which the variance of the targeted module in the immersed network setting is higher compared to the full-MISO approach. Other approaches and measures to analyse variance, such as the geometric approach (\cite{hjalmarsson2011geometric}), frequency response variance expression for finite order and sample size (\cite{hjalmarsson2004exact}), or transfer function variance expression for finite order (\cite{ninness2004variance}) may also be used, but are not considered here. In a case study, we consider a four-node dynamic network and show how the signal-to-noise ratio and the number/complexity of the to-be-identified transfer functions affect the quality of estimate in terms of variance.

The paper is organised as follows: in the next section, preliminaries such as the definition of the dynamic network and variance measures are discussed. The main theoretical results for the variance analysis are given in section \ref{sec:variance_analysis}. Later, in section \ref{sec:case_study}, a case study with a three-node dynamic network is presented. Conclusions of the work are drawn afterwards.

\section{Preliminaries and system setup}
\subsection{Dynamic network}
We consider the dynamic network setup as given in \cite{van2013identification} and \cite{dankers2016identification}. The network is built up of $L$ nodes, related to $L$ scalar internal variables $w_j, j=1,\cdots,L$. It is assumed that each internal variable can be written as:
\begin{align} \label{eq:original_network}
  w_j(t) &= \sum_{k \in \mathcal{N}_j} G_{jk}^0(q) w_k(t) + r_j(t) + v_j(t),
\end{align}
where $G_{jk}^0(q), k \in \mathcal{N}_j$ is a proper rational transfer function, $q^{-1}$ is the delay operator, i.e., $q^{-1}u(t) = u(t-1)$ and,
\begin{itemize}
  \item $\mathcal{N}_j$ is the set of indices of internal variables with direct causal connection to $w_j, i.e., i \in \mathcal{N}_j$ iff $G_{ji}^0 \neq 0$;
  \item $v_j$ is an unmeasured disturbance variable with rational spectral density: $v_j = H_j^{0}(q)e_j$, where $e_j$ is a white noise process, and $H_j^{0}$ is a monic, stable, and minimum phase transfer function;
  \item $r_j$ is an external variable that is known and, if it is present at node $j$, it can be manipulated by the user.
\end{itemize}
In the matrix form, all measured variables can be written as:
\begin{align*}
  \left(
    \begin{smallmatrix}
      w_1 \\
      w_2 \\
      \vdots \\
      w_L \\
    \end{smallmatrix}
  \right) &=
  \left(
    \begin{smallmatrix}
      0 & G_{12}^0 & \cdots & G_{1L}^0 \\
      G_{21}^0 & 0 & \ddots & G_{2L}^0 \\
      \vdots & \ddots & \ddots & \vdots \\
      G_{L1}^0 & G_{L2}^0 & \cdots & 0 \\
    \end{smallmatrix}
  \right)
  \left(
  \begin{smallmatrix}
    w_1 \\
    w_2 \\
    \vdots \\
    w_L
  \end{smallmatrix} \right) +
  \left(
  \begin{smallmatrix}
      r_1 \\
    r_2 \\
    \vdots \\
    r_L
  \end{smallmatrix}\right) +
   \left(
   \begin{smallmatrix}
      v_1 \\
    v_2 \\
    \vdots \\
    v_L
  \end{smallmatrix}\right) \\
  &= G^0w + r + v = (I - G^0)^{-1}(r+v).
\end{align*}
In \cite{van2013identification}, it is shown that for a consistent estimate of a local module in a dynamic network $G_{jk}^0$, it is sufficient to consider all inputs $k \in \mathcal{N}_j$ in the set of predictor inputs. This condition is relaxed in the immersed network setting (\cite{dankers2016identification}), where consistency is achieved with a reduced number of signals as predictor inputs. The consistent estimate of $G_{jk}^0$ is possible if:
\begin{itemize}
  \item $w_k$ is included as predictor input;
  \item each parallel path from $w_k \rightarrow w_j$ passes through a node is chosen as predictor input; and
  \item each loop from $w_j \rightarrow w_j$ passes through a node is also chosen as predictor input,
\end{itemize}
while for the so-called direct method of identification a condition has to be added on the absence of confounding variables.

In the immersed network settings the new model identification equation becomes:
\begin{align} \label{eq:immersed_network_tf}
  w_j(t) &= \sum_{k \in \mathcal{D}_j} \breve{G}_{jk}(q) w_k(t) + r_j + \breve{v}_j,
\end{align}
where $\breve{\cdot}$ indicates terms corresponding to the immersed network, $\mathcal{D}_j$ denote the set of indices of the internal variables chosen as predictor inputs, $\breve{G}_{jk}$ are identified models in immersed network and they can be different from $G_{jk}^0$, except for the targeted module and the noise $\breve v_j(t)$ can be defined as $\breve v_j(t) = v_j(t) + \sum_{k \in N_j \backslash D_j} G_{jk}w_k(t)$. Let, for a targeted local module $G_{j1}^0$, $M\in\mathbb{N}$  be the number of inputs that can be removed from the predictor input set, while still guaranteeing consistency, resulting in $L_{j}-M$ inputs in immersed settings, wherein $L_j$ is the number of nodes connected to node $j$.

\subsection{Asymptotic transfer function covariance expression}
In \cite{ljung1985asymptotic}, an expression for the estimated black-box transfer function, being asymptotic in both the number of to-be-estimated parameters of the transfer function $n$ and data points $N$, is derived and reads:
\begin{align}\label{eq:tf_expression}
  \text{cov} \left(\begin{array}{c}
       \hat{G}_N(e^{i\omega},n) \\
       \hat{H}_N(e^{i\omega},n)
     \end{array}\right) \approx
     \frac{n}{N} \Phi_v(\omega) \cdot
     \left(
       \begin{array}{cc}
         \Phi_u(\omega) & \Phi_{ue}(\omega) \\
         \Phi_{ue}(-\omega) & \lambda \\
       \end{array}
     \right)^{-1},
\end{align}
where $\Phi_{u}$ and $\Phi_v$ are respectively the input and noise spectrum, where the latter has a noise variance $\lambda$; $\Phi_{ue}$ is the cross spectrum between input $u$ and the noise signal $e$, and $\hat{G}_N$ and $\hat{H}_N$ are respectively the estimate of the true transfer functions $G_0$ and $H_0$.

\subsection{Parameter covariance matrix}
In a prediction error identification setting, for a parametric model $G(q,\theta)$, the covariance matrix $P_\theta$ is the inverse of the \textit{Fisher information matrix} $\mathcal{M}$ and reads: (\cite{ljung1998system}):
\begin{align} \label{eq:pthetapar}
  P_\theta =\sigma_e^2[E \psi(t,\theta) \psi(t,\theta)^\top]^{-1} = \sigma_e^2 \mathcal{M}^{-1},
\end{align}
where $\sigma_e^2$ is the noise variance, $E$ is the time expected value and $\psi(t,\theta)$ is the gradient of the one-step ahead prediction error $\epsilon(t/t-1,\theta)$ with respect to the parameters (i.e., the sensitivity of these errors to parameter variations):
\begin{align*}
  \psi(t,\theta) := \frac{\partial \epsilon(t/t-1,\theta)}{\partial \theta} = - \frac{\partial \hat{y}(t/t-1,\theta)}{\partial \theta},
\end{align*}
where $\hat{y}(t/t-1,\theta)$ is one-step ahead predictor.

In the next section, theoretical conditions are derived for the variance analysis of the full-MISO and immersed network approaches.

\section{Variance analysis} \label{sec:variance_analysis}
In a large network with hundreds of nodes, it is important to analyse which information (inputs) is most relevant for the estimation of a targeted module, e.g., $G_{j1}(q)$. We ponder the question: is it always better to consider \textit{all} possible information sources (inputs) in the modelling process of a targeted module or, instead, only a subset of all inputs? In other words,
while still guaranteeing consistency, can we, in the immersed
settings, provide a more precise estimate? We address this
question below, where we use the asymptotic transfer function variance expression as given in eq. \eqref{eq:tf_expression}.

\subsection{Variance analysis based on transfer function aymptotic variance expression } \label{sec:tf_variance}
We use the expression of variance as in eq. \eqref{eq:tf_expression} for the case of MISO system with $L_j$ inputs. Let the targeted module be $G_{j1}^0$ and define the $(L_j+1)\times 1$ vector of input signals
\begin{equation}
\label{eq:wvec}
x(t) = [w_{1}(t) \ldots e_j(t) \ldots w_{L_j+1}(t)]^{\top}.	
\end{equation}
Furthermore, let
\begin{equation}
	\label{eq:g_general}
	\mathcal{G}_{\gamma}=[\hat{G}_{j1}(e^{i\omega},n), H_{j}(e^{i\omega},n), \ldots, \hat{G}_{j\gamma}(e^{i\omega},n)]^{\top},
\end{equation}
where $\gamma\in\mathbb{N}_{\geq 1}$ indicates the number of inputs. Then the asymptotic transfer matrix covariance expression is given by:
\begin{align}\label{eq:tf_expression_MISO}
  \text{cov}\left(\mathcal{G}_{\gamma}\right) \approx
     \frac{n_{\gamma}}{N} \Phi_v(\omega) \cdot
\left[\sum_{\tau=-\infty}^{\infty}e^{-i\omega \tau }\bar{E} x(t) x^{T}(t-\tau)\right]^{-1},
\end{align}
where $\bar{E}$ is the expectation over time.
Using a block-matrix inverse, the covariance of the estimated transfer function of interest, $\hat{G}_{j1}(e^{i\omega},n)$, with $\gamma = L_j$ inputs is computed from (\ref{eq:tf_expression_MISO}) and reads:
\begin{align}\label{eq:covG_21_L}
  \text{cov}(\hat{G}_{j1}(e^{i\omega},n)) &\approx \frac{n_{L_j}}{N} \Phi_v(\omega) \cdot
  \left[
  \Phi_{w_{1}}(\omega) - \Gamma_{L_j} \Upsilon_{L_j}^{-1} \Gamma_{L_j}^{H}
  \right]^{-1},
\end{align}
where $(\cdot)H$ stands for the Hermitian conjugate,
\begin{align}
\label{eq:gamma}
  \Gamma_\gamma &= \left(\begin{array}{cccc}
    \Phi_{w_{(1)}e}(\omega)   & \Phi_{w_{(1)}w_{3}}(\omega)   & \cdots & \Phi_{w_{(1)}w_{\gamma}}(\omega)
  \end{array} \right),
 \end{align}
 and
\begin{align}
\label{eq:ups}
  \Upsilon_\gamma &=
  \left(
    \begin{array}{cccc}
      \lambda                 & \Phi_{e w_{3}}(\omega)      & \cdots & \Phi_{e w_{\gamma}}(\omega)\\
      \Phi_{e w_{3}}(-\omega) & \Phi_{w_{3}}(\omega)        & \cdots & \Phi_{w_{3}w_{\gamma}}(\omega)\\
      \vdots                  & \vdots                      & \ddots & \vdots \\
      \Phi_{e w_{\gamma}}(-\omega) & \Phi_{w_{3}w_{\gamma}}(-\omega)   & \cdots & \Phi_{w_{\gamma}w_{\gamma}}(\omega)
    \end{array}
  \right).
\end{align}
We can perform the same calculations for the immersed network that has a reduced set of $L_j-M$ inputs. Using the breve notation as applied in (\ref{eq:immersed_network_tf}), we define the vector of transfer functions $\breve{\mathcal{G}}_{\gamma}$ (cf. \ref{eq:g_general}). Expression (\ref{eq:tf_expression_MISO}) and (\ref{eq:covG_21_L}) then follow straightforwardly for the case of the immersed network, where the latter becomes:
\begin{align}\label{eq:covG_21_L-m}
  \text{cov}(\hat{G}_{j1}(e^{i\omega},n))\approx\,\,\,\,\,\,\,\,\,\,\,\,\,\,\,\,\,\,\,\,\,\,\,\,\,\,\,\,\,\,\,\,\,\,\,\,\,\,\,\,\,\,\,\,\,\,\,\,\,\,\,\,\,\,\,\,\,\,\,\,\,\,\,\,\,\,\,\,\,\,\,\,\,\,\,\,\,\,\,\,\,\,\,\,\,\,\,\,\,\,\,\,\,\,\,\,\,\,\,\,\,\, \notag\\
  \frac{n_{L_j-M}}{N} \Phi_{\breve{v}}(\omega) \cdot
  \left[
  \Phi_{w_{1}}(\omega) - \breve{\Gamma}_{L_j-M} \breve{\Upsilon}_{L_j-M}^{-1} \breve{\Gamma}_{L_j-M}^{H}
  \right]^{-1},
\end{align}
in which $\breve{\Gamma}$ and $\breve{\Upsilon}$ are the immersed network equivalents of respectively (\ref{eq:gamma}) and (\ref{eq:ups}) (the noise $e(t)$ is replaced with it's respective immersed network counterparts $\breve{e}(t)$).

\subsection{Main results}

\begin{thm} \label{th:tf_exp_1}
If
\begin{equation} \label{eq:th2_condition}
\resizebox{1\hsize}{!}{
  $\frac{\Phi_{\breve{v}}(\omega)}{\Phi_v(\omega)} >
   \frac{n_{L_j}\left[\Phi_{w_{(1)}}(\omega) - \breve{\Gamma}_{(L_j-M)}(\omega) \breve{\Upsilon}_{(L_j-M)}^{-1}(\omega) \breve{\Gamma}_{(L_j-M)}^H(\omega)\right]}{n_{L_j-M}\left[\Phi_{w_{(1)}}(\omega) - \Gamma_{L_j}(\omega) \Upsilon_{L_j}^{-1}(\omega) \Gamma_{L_j}^H(\omega)\right]}$}
\end{equation}

then $\text{cov}^{(L_j-M)}(\hat{G}_{j1}(e^{i\omega},n)) > \text{cov}^{(L_j)}(\hat{G}_{j1}(e^{i\omega},n))$.
\end{thm}
\begin{proof}
See Appendix.
\end{proof}

Theorem \ref{th:tf_exp_1} provides a result for each frequency in the frequency response of the targeted module.

\subsection{Variance analysis based on parameter covariance expression}

As an alternative for the frequency-based variance analysis, we briefly cover variance analysis based on a scalar measure on the covariance matrix of the to-be-identified parameters. The covariance matrices of both original and immersed networks are respectively given by $P_{\alpha}^{L_j}$ and $P_{\alpha}^{L_j-m}$; see \eqref{eq:pthetapar}. Below, we will consider two useful measures: the $E$- and $D$-optimality measures.
\subsubsection{E-optimality and D-optimality}
The $D$-optimality criterion measures the informativeness of an experiment based on the volume of the confidence ellipsoid that can be constructed with the covariance matrix of the to-be-identified parameters \eqref{eq:pthetapar}, i.e., when
\begin{equation}
	\det((P_{\theta}^{L_j})^{-1}) > \det((P_{\theta}^{L_j-M})^{-1}),
\end{equation}
the overall accuracy of the parameters in the module $G_{j1}$ is higher than those in the immersed network. $E$-optimality can indicate whether one ellipsoid lies completely inside the other. Mathematically, this conditions holds when
\begin{equation}
	(P_{\theta}^{L_j})^{-1}\succeq (P_{\theta}^{L_j-M})^{-1},
\end{equation}
Hence, contrasting the $D$-optimality measure, all variances of the to-be-identified parameters in the full-MISO approach are smaller than in its equivalent immersed network representation if the above condition is satisfied. 

\section{Case Study} \label{sec:case_study}
To study the effect of the number and complexity of the to-be-estimated transfer functions and the poor signal-to-noise ratio on the variance of a targeted module, we consider a four-node dynamic network as shown in Fig. \ref{fig:example_network}.
\begin{figure} [h]
\center
    \includegraphics[width=3in]{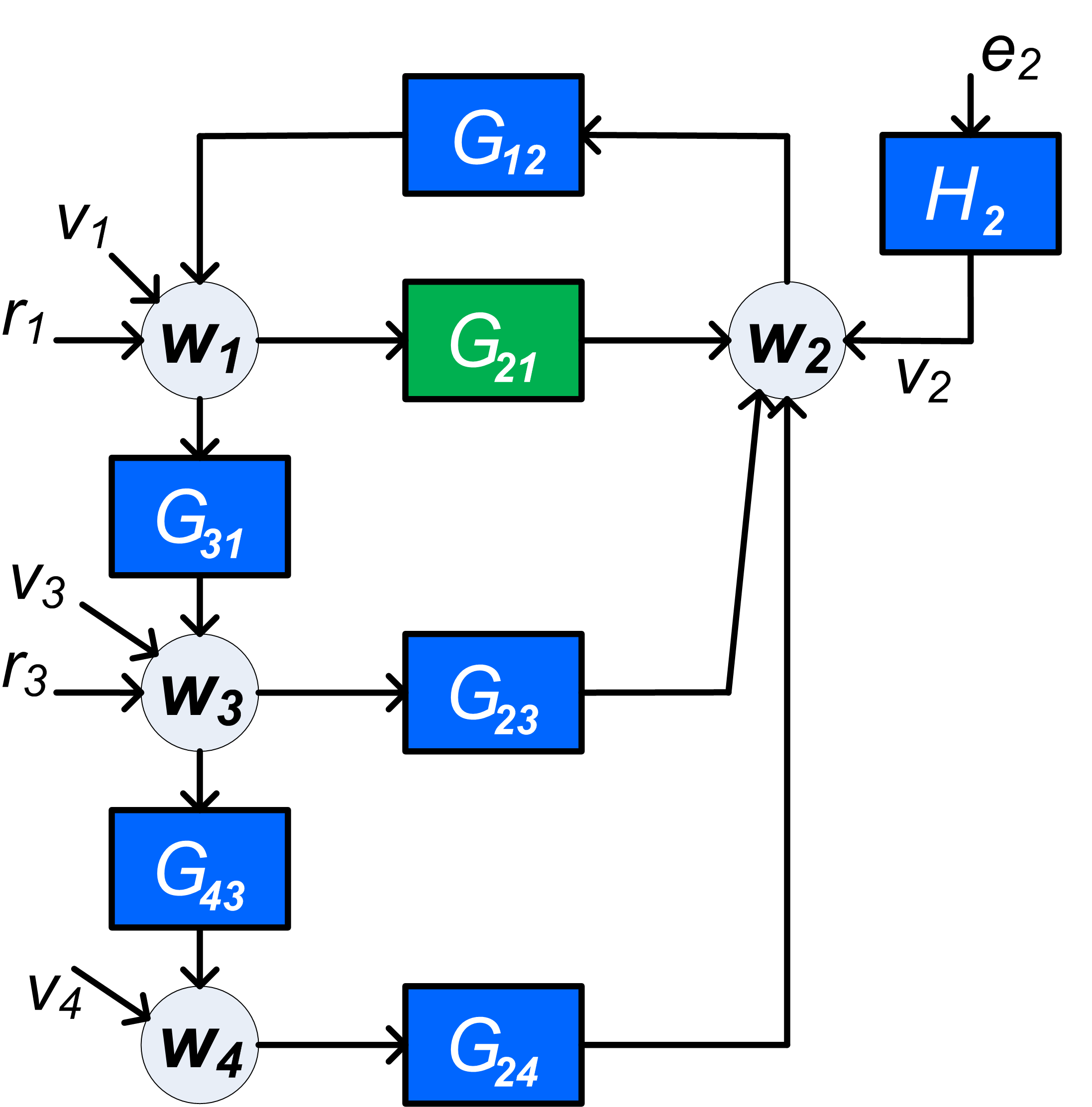}
  \caption{4-node dynamic network}  \label{fig:example_network}
\end{figure}
The dynamic network representation is given by:
\begin{align} \label{eq:general_equation}
w = G w + r + v,
\end{align}
In the matrix form,
\begin{align}
\left(
  \begin{array}{c}
    w_1 \\
    w_2 \\
    w_3 \\
    w_4 \\
  \end{array}
\right)
&=
\left(
  \begin{smallmatrix}
    0 & G_{12} & 0 & 0 \\
    G_{21} & 0 & G_{23} & G_{24} \\
    G_{31} & 0 & 0 & 0 \\
    0 & 0 &  G_{43} & 0 \\
 \end{smallmatrix}
\right)
\left(
  \begin{array}{c}
    w_1 \\
    w_2 \\
    w_3 \\
    w_4 \\
  \end{array}
\right)
+
\left(
  \begin{array}{c}
    r_1 \\
    0 \\
    r_3 \\
    0 \\
  \end{array}
\right)
+
\left(
  \begin{array}{c}
    v_1 \\
    v_2 \\
    v_3 \\
    v_4 \\
  \end{array}
\right)
\end{align}
In the example, the targeted module is $G_{21}$. For the direct identification method as explained in \cite{van2013identification}, the full-MISO approach requires $\{ w_1,w_3,w_4\}$ as predictor inputs to obtain a consistent estimate of $G_{21}$, as shown in the equation below.
\begin{align}
 w_2 = G_{21} w_1 + G_{23} w_3  + G_{24} w_4 + H_2 e.
\end{align}
Using the conditions presented in \cite{dankers2016identification}, a consistent estimate of $G_{21}$ is obtained with a reduced number of predictor inputs $\{ w_1,w_3\}$, i.e., $w_4$ can be removed. The immersed network is shown in Fig. \ref{fig:immersed_network}.
\begin{figure} [h]
\center
    \includegraphics[width=3in]{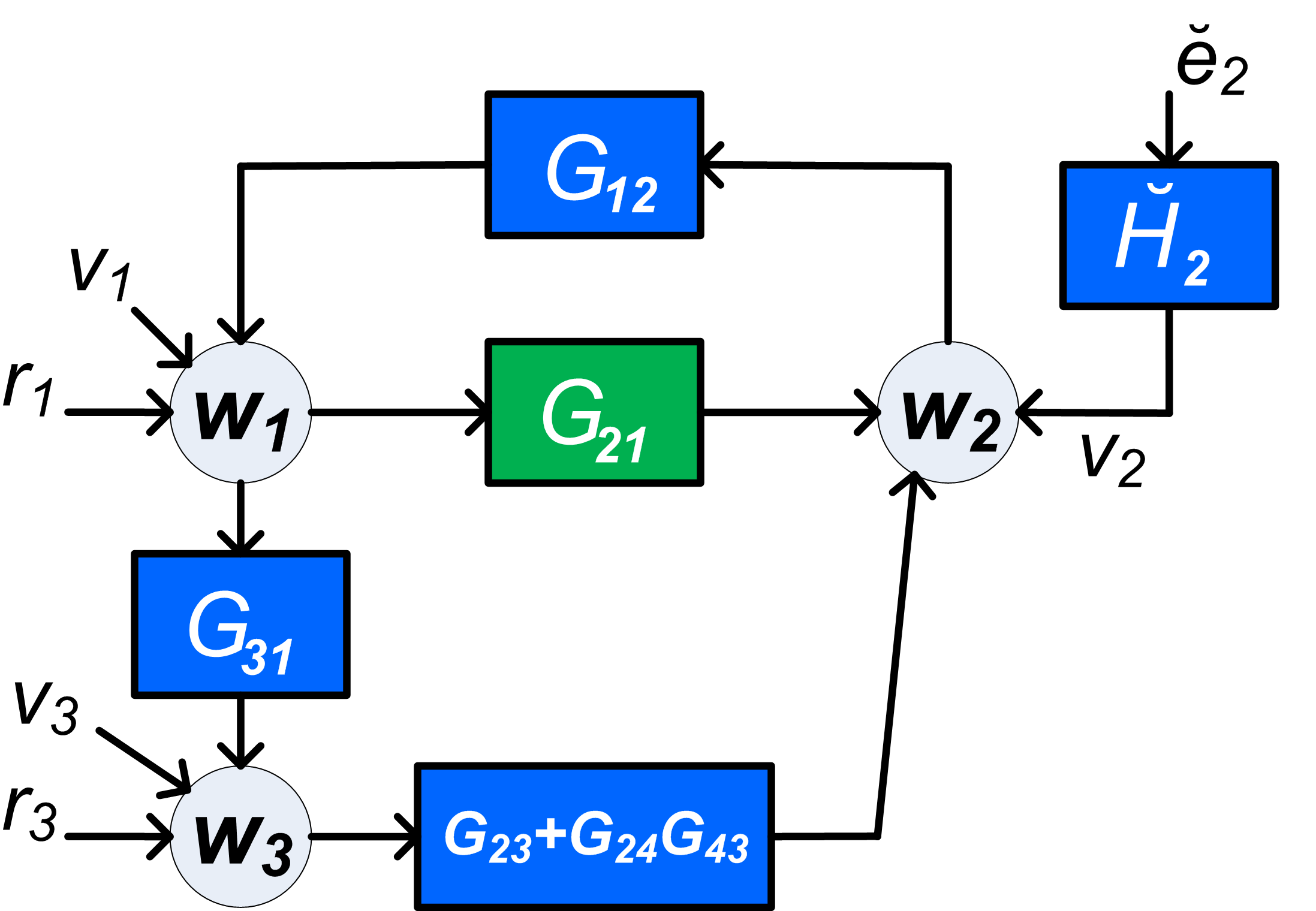}
  \caption{Immersed dynamic network with reduced number of nodes}  \label{fig:immersed_network}
\end{figure}
In the immersed network settings, the following equation is derived for the estimation of $G_{21}$:
\begin{align} \label{eq:immersed_network}
  w_2(t) &= G_{21}(q) w_1(t) + \breve{G}_{23}(q) w_3(t) +  \breve{H}_2(q) \breve{e}(t),
\end{align}
where $\breve{G}_{23} = G_{23} + G_{24} G_{43}$ and $\breve{H}_2$ can be obtained via a spectral decomposition, see e.g., \cite{dankers2016identification}.

Two aspects, i.e. poor signal-to-noise ratio and the number/complexity of the to-be-estimated transfer functions in the immersed network setting, affecting the quality of estimate of $G_{21}^0$ are discussed next:
\subsection{Poor signal-to-noise ratio in immersed network}
Process noise $v$ at input nodes can provide external excitation in dynamic network. For the considered example, the noise $v_4$ on the node signal $w_4$ can induce excitation when the node signal $w_4$ is considered in the set of predictor inputs (full-MISO case). In the immersed network settings, when $w_4$ is removed from the set of predictor inputs, the noise $v_4$ directly affects node $w_2$, causing a poor signal-to-noise ratio at the output. The important point to consider is that the transfer function $G_{24}$ determines how the noise $v_4$ affects output node $w_2$. A low gain $G_{24}$ results in a minimal effect while a high gain transfer function $G_{24}$ worsens the signal-to-noise ratio. Therefore, the gain of $G_{24}$ becomes an important variable to consider for the variance analysis of targeted module $G_{21}$.

\subsection{Number and/or complexity of to-be-estimated transfer functions in immersed network}
Another important aspect is the number/complexity of to-be-estimated transfer functions. It can be observed that in the immersed network, the number of transfer functions is reduced compared to full-MISO approach. In the considered example, only two transfer functions, $G_{21}$ and $\breve{G}_{23}$ are to be estimated. However, $\breve{G}_{23}$ can be more complex than $G_{23}$. For example, $\breve{G}_{23} = G_{23} + G_{24} G_{43}$ becomes complex if the number of parameters in $G_{43}$ are high or if there is no pole-zero cancellation in $G_{24}$ and $G_{43}$ and hence there are more to-be-estimated parameters in the immersed network settings compared to the full-MISO approach. To study this effect on the variance of the targeted module $G_{21}$, we consider $G_{43}$ with different orders.

In the following simulation example, we analyse the effect of i)- the gain of $G_{24}$ and ii)- the complexity of $G_{43}$ on the quality of estimation of $G_{21}$ in terms of variance.

\subsection{Simulation experiment}
For the variance analysis, we only consider the conditions derived for the asymptotic transfer function covariance expression and compare the results with the sample covariance measure. It would be attractive if the theoretical analysis could lead to generic conclusions but, due to the complexity of the problem at hand, it is difficult to achieve. As an alternative these theoretical results are used to a posteriori compare the variance expressions in a particular case. The simulation details and results are presented in the next subsections.

\subsection{Experiment details}
We have implemented the network shown in Fig. \ref{fig:example_network}. A white noise signal with power $0.1$ is used as external excitation signals, i.e., $\{r_1, r_3\}$ and is applied at nodes $1$ and $3$. All the white noises $e$ have variance $0.1$, i.e., $\lambda = 0.1$. The local module of interest is $G_{21}$. An (simulation) experiment is performed and the data set of length $N = 10^{4}$ with sampling time $T_{s} = 1$sec is collected. In the first modelling process setup, direct identification for dynamic networks (\cite{van2013identification} is performed using all inputs (full-MISO approach), i.e., $\{w_1,w_3,w_4\}$. A second modelling setup is considered with immersed network (\cite{dankers2016identification}) where the predictor inputs include $\{w_1,w_3\}$. Consistency results for both modellings setups are presented in \cite{van2013identification} and \cite{dankers2016identification}. A Box-Jenkins (BJ) model structure is used for the identification of each module $G_{21},G_{23},G_{43}, \breve{G}_{23}$ and noise filters $H_2, \breve{H}_2$ in both setups.

\subsection{Results}
Firstly, the transfer function $G_{43}$ is considered to have only one parameter, i.e., delay with a gain. One hundred Monte-Carlo simulations are performed for different gains of $G_{24}$, i.e., $\{0.005, 0.05, 0.5, 1\}$. Fig. \ref{fig:sample_cov_G_43_one_parameter} shows sample covariance per frequency over the $100$ Monte-Carlo simulations, where the sample covariance is defined as:
\begin{align*}
  \text{cov} (G_{j1}(e^{i\omega},n)  = \frac{1}{\mathbf{n}} \sum_{i=1}^{\mathbf{n}} ((G_{j1}^i(e^{i\omega},n) - E[(G_{j1}^i(e^{i\omega},n)])^2,
\end{align*}
where $\mathbf{n}$ is the number of Monte-Carlo simulations and $E$ represents the average value.
\begin{figure} [h]
\center
    \includegraphics[width=3in]{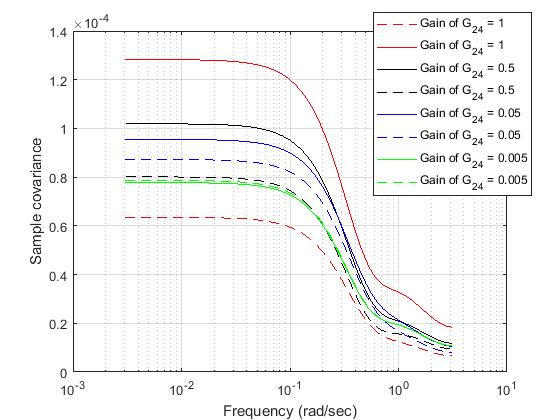}
  \caption{Sample covariance per frequency of $G_{21}$ when $G_{43}$ has only one parameter. In this case, the number of to-be-estimated parameters in the immersed $n_{L_j-M}$ and the full-MISO approaches $n_{L_j}$ are the same, dotted line shows full-MISO approach, solid line shows immersed network}  \label{fig:sample_cov_G_43_one_parameter}
\end{figure}

It can be observed in Fig. \ref{fig:sample_cov_G_43_one_parameter} that as the gain of $G_{24}$ is reduced the covariance of $G_{21}$ reduces, as expected. For $0.005$, the covariance of the immersed network becomes smaller compared to the full-MISO approach, resulting in a better estimate even with a reduced number of information sources (inputs) used in the estimation.

Fig. \ref{fig:sample_cov_G_43_two_parameter} shows the covariance per frequency when $G_{43}$ has two parameters. In this case, the complexity of the to-be-identified transfer functions in the immersed network settings is higher compared to the full-MISO approach. The first observation is that, as can also be observed in fig. \ref{fig:sample_cov_G_43_one_parameter}, the covariance of $G_{21}$ reduces with the gain of $G_{24}$.  Due to an increased complexity the covariance of $G_{21}$, in this case, is always higher in the immersed network settings irrespective of the gain of $G_{24}$.
\begin{figure} [h]
\center
    \includegraphics[width=3in]{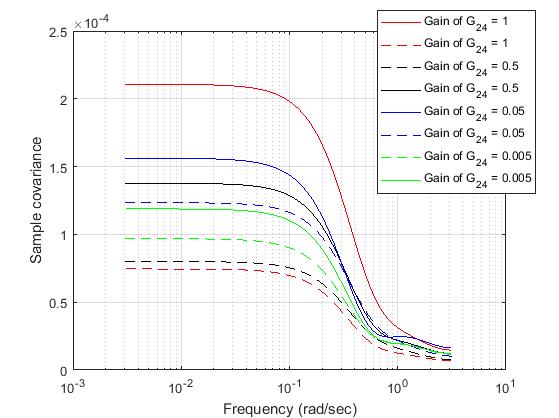}
  \caption{Sample covariance per frequency of $G_{21}$ when $G_{43}$ has two parameters. In this case, the number of to-be-estimated parameters in the immersed settings are more than the full-MISO approach, i.e., $n_{L_j-M} > n_{L_j}$, dotted line shows full-MISO approach, solid line shows immersed network}  \label{fig:sample_cov_G_43_two_parameter}
\end{figure}

The theoretical conditions using the asymptotic transfer function expression as derived in section \ref{sec:tf_variance} are also implemented. Fig \ref{fig:conditions_one_parameter} and \ref{fig:conditions_two_parameter} show the conditions when $G_{43}$ has one and two parameters respectively. For the case when $G_{43}$ has one parameter and with gain $0.005$ for $G_{24}$, the condition becomes negative, indicating that the covariance of full-MISO approach is higher compared to the immersed network settings. The theoretical results confirms the observation obtained via the sample covariance.
\begin{figure} [h]
\center
    \includegraphics[width=3in]{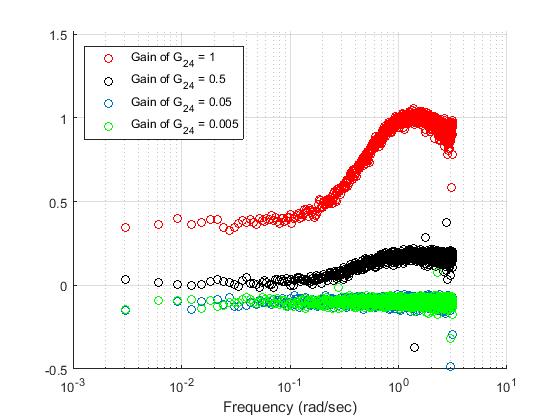}
  \caption{An implementation of the condition   $\frac{\Phi_{\breve{v}}(\omega)}{\Phi_v(\omega)} -
   \frac{n_L.\Phi_{w_{(1)}}(\omega) - \breve{\Gamma}_{(L_j-M)}(\omega) \breve{\Upsilon}_{(L_j-M)}^{-1}(\omega) \breve{\Gamma}_{(L_j-M)}^H(\omega)}{n_{L_j-M}.\Phi_{w_{(1)}}(\omega) - \Gamma_{L_j}(\omega) \Upsilon_{L_j}^{-1}(\omega) \Gamma_{L_j}^H(\omega)}$ per frequency, When $G_{43}$ has one parameter. According to theorem 1, if $> 0$, then it implies covariance of the immersed network is higher compared to the full-MISO approach}  \label{fig:conditions_one_parameter}
\end{figure}

\begin{figure} [h]
\center
    \includegraphics[width=3in]{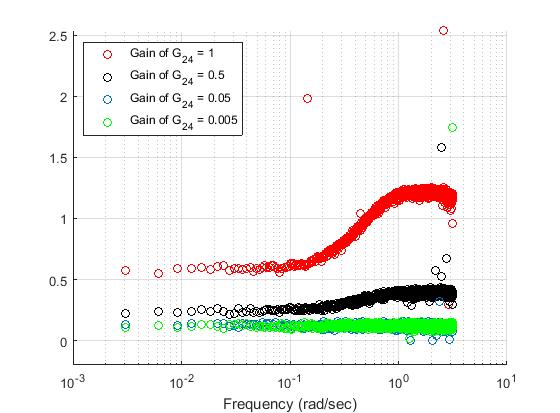}
  \caption{An implementation of the condition   $\frac{\Phi_{\breve{v}}(\omega)}{\Phi_v(\omega)} -
   \frac{n_L.\Phi_{w_{(1)}}(\omega) - \breve{\Gamma}_{(L_j-M)}(\omega) \breve{\Upsilon}_{(L_j-M)}^{-1}(\omega) \breve{\Gamma}_{(L_j-M)}^H(\omega)}{n_{L_j-M}.\Phi_{w_{(1)}}(\omega) - \Gamma_{L_j}(\omega) \Upsilon_{L_j}^{-1}(\omega) \Gamma_{L_j}^H(\omega)}$ per frequency, When $G_{43}$ has two parameter. According to theorem 1, if $> 0$, then it implies covariance of the immersed network is higher compared to the full-MISO approach}  \label{fig:conditions_two_parameter}
\end{figure}%

\section{Conclusions}
We analysed the quality of a local module estimate in dynamic networks, in terms of its variance, for two different modelling setups: the full-MISO and immersed networks. The research question addresses the effect of using an additional input for the estimation problem. The variance analysis for a targeted module in dynamic networks is a complex problem where many factors, e.g. the dynamics of the targeted and other transfers in the network, different excitation conditions etc., can play an important role. We identify two dominant factors that affect the quality of estimate and it is shown that it is possible to obtain a better quality in terms of a reduced variance with smaller set of inputs compared to considering all inputs as predictor inputs.


\bibliography{MyCollection}             

\appendix
\section{Appendix}    

\subsubsection{Proof theorem \ref{th:tf_exp_1}:}
We need to show that:
\begin{align*}
 \text{cov}^{(L_j-M)}(\hat{G}_{j1}(e^{i\omega},n)) > \text{cov}^{(L_j)}(\hat{G}_{j1}(e^{i\omega},n)),
 \end{align*}
which can also be written as:
\begin{align}\label{eq:ineq2_to_proof}
  \text{cov}^{(L_j-M)}(\hat{G}_{j1}(e^{i\omega},n)) - \text{cov}^{(L_j)}(\hat{G}_{j1}(e^{i\omega},n)) > 0
\end{align}
Substituting eq. \eqref{eq:covG_21_L-m} and \eqref{eq:covG_21_L} in above eq. \eqref{eq:ineq2_to_proof}, we obtain
\begin{align*}
  \frac{n_{L_j-M}}{N} \Phi_{\breve{v}}(\omega) \cdot
  \left(
  \Phi_{w_{(1)}}(\omega) - \breve{\Gamma}_{L_j-M} \breve{\Upsilon}_{L_j-M}^{-1} \breve{\Gamma}_{L_j-M}^H
  \right)^{-1} -\\
  \frac{n_{L_j}}{N} \Phi_v^{(L_j)}(\omega) \cdot
  \left(
  \Phi_{w_{1}}(\omega) - \Gamma_{L_j} \Upsilon_{L_j}^{-1} \Gamma_{L_j}^H
  \right)^{-1}> 0.
\end{align*}
Dividing the above expression by $\frac{N}{ {\Phi_v^{(L_j)}(\omega)}}$, we obtain
\begin{align*}
    \frac{\Phi_{\breve{v}}(\omega)}{\Phi_v(\omega)}  \cdot
   n_{L_j-M}\left(
  \Phi_{w_{(1)}}(\omega) - \breve{\Gamma}_{L_j-M} \breve{\Upsilon}_{L_j-M}^{-1} \breve{\Gamma}_{L_j-M}^H
  \right)^{-1} -\\
     n_{L_j}\left(
  \Phi_{w_{(1)}}(\omega) - \Gamma_{L_j} \Upsilon_{L_j}^{-1} \Gamma_{L_j}^H
  \right)^{-1}> 0,
  \end{align*}
which can be rearranged into the inequality
\begin{equation}
\resizebox{1\hsize}{!}{
  $\frac{\Phi_{\breve{v}}(\omega)}{\Phi_v(\omega)} >
   \frac{n_{L_j}\left[\Phi_{w_{(1)}}(\omega) - \breve{\Gamma}_{(L_j-M)}(\omega) \breve{\Upsilon}_{(L_j-M)}^{-1}(\omega) \breve{\Gamma}_{(L_j-M)}^H(\omega)\right]}{n_{L_j-M}\left[\Phi_{w_{(1)}}(\omega) - \Gamma_{L_j}(\omega) \Upsilon_{L_j}^{-1}(\omega) \Gamma_{L_j}^H(\omega)\right]}$}
\end{equation}

\end{document}